\documentclass[a4paper, 10 pt, conference]{IEEEtran}
\usepackage{amsmath,nccmath}
\usepackage{amssymb}
\usepackage{lipsum}
\usepackage{graphicx}
\usepackage{esint}
\usepackage{cite}
\usepackage{amsmath,amssymb,amsthm}
\usepackage{thmtools,thm-restate}
\usepackage{hyperref}
\usepackage{cleveref}
\usepackage{ifpdf}
\usepackage{color}
\usepackage{placeins}
\usepackage{float}
\usepackage{tabularx}
\usepackage{colortbl}
\usepackage{pgfplots}
\usepackage{tikz}
\usepackage{tikzscale}
\pgfplotsset{compat=newest}
\usetikzlibrary{plotmarks}
\usetikzlibrary{arrows.meta}
\usepgfplotslibrary{patchplots}
\usepackage{grffile}
\pgfplotsset{plot coordinates/math parser=false}
\newlength\figureheight
\newlength\figurewidth
\usepackage{pgfgantt}
\usepackage{pdflscape}

\begin{document}

\title{Guaranteed Evader Detection in Multi-Agent
Search Tasks using Pincer Trajectories}


\author{Roee M. Francos
        and Alfred M. Bruckstein
\thanks{Roee M. Francos and Alfred M. Bruckstein are with the Faculty
of Computer Science, Technion- Israel Institute of Technology, Haifa, Israel, 320003, Emails:(roee.francos@cs.technion.ac.il,
  alfred.bruckstein@cs.technion.ac.il).}
}

\maketitle
\begin{abstract}
Assume that inside an initial planar area there are smart mobile evaders attempting to avoid detection by a team of sweeping searching agents. All sweepers detect evaders with fan-shaped sensors, modeling the field of view of real cameras. Detection of all evaders is guaranteed with cooperative sweeping strategies, by setting requirements on sweepers' speed, and by carefully designing their trajectories. Assume the smart evaders have an upper limit on their speed which is a-priori known to the sweeping team. An easier task for the team of sweepers is to confine evaders to the domain in which they are initially located. The sweepers accomplish the confinement task if they move sufficiently fast and detect evaders by applying an appropriate search strategy. Any given search strategy results in a minimal sweeper's speed in order to be able to detect all evaders. The minimal speed guarantees the ability of the sweeping team to confine evaders to their original domain, and if the sweepers move faster they are able to  detect all evaders that are present in the region. 
We present results on the total search time for a novel pincer-movement based search protocol that utilizes complementary trajectories along with adaptive sensor geometries for any even number of pursuers.
\end{abstract}


\IEEEpeerreviewmaketitle

\section{Introduction}
\textbf{Motivation}. The goal of this research is to provide efficient ``must-win" search strategies for a team of $n$ identical sweeping agents that must guarantee detection of an unknown number of smart evaders that are initially located inside a given circular region of radius $R_0$ while minimizing the search time. The evaders move and try to escape the initial region and have maximal speed of $V_T$, known to the sweepers, and do not have any turning constraints.

A smart evader is one that detects and responds to the motions of searchers by performing optimal evasive maneuvers, to avoid interception. A smart evader is assumed to have full knowledge of the search strategy of the sweeping team. Possessing such knowledge enables smart evaders to plan their movements in a way that maximizes the time required for the pursuing sweepers to detect them. Guaranteed detection of all evaders implies that for all particular choices of escape strategies smart evaders may implement in response to the motions of searches, they will all eventually be detected. 

Agents gather information only from their sensors, and evaders in a sweeper's field of view are immediately detected. There can be many evaders, and those can be located at any point in the interior of the circular region at the beginning of the search process. Importantly, in contradiction to most recent literature on pursuit-evasion problems, the sweepers do not have any information regarding evaders locations' that are outside of their sensing range, nor on the total number of evaders they must detect. All sweepers move at a speed $V_s > V_T$ (measured at the center of the sensor that represents a sweeper's center of mass) and detect evaders using fan shaped sensors with a given half-angle denoted by $\alpha$ and a length of $2r$. Finding an efficient algorithm requires that, throughout the sweep, the footprint of the sweepers’ sensors maximally overlaps the evader region (the region where evaders may possibly be) in order to detect as many evaders as possible. 

The sweeping team consists of an even number of agents, referred to as sweepers, that act as sensors and sweep the region until all evaders are detected. The search is done by pairs of agents sweeping toward each other for the purpose of entrapping all evaders. Cooperation among agents scanning adjacent sectors and sweep toward each other enables them to entrap all evaders, regardless of evaders' adversarial motions.

Since the sweepers do not have any additional knowledge about the evaders whereabouts, or even if all evaders were found at some intermediate point of time during the search, the search is continued until the whole region is searched. 

Additional information such as the number of evaders and their exact locations provided to the sweeper team in advance, could reduce the termination time by utilizing this knowledge and organizing the entire search process differently. However this is not the focus of this paper. Therefore, the resulting search times for a circular environment can be seen as an upper bound on the search time, resulting from the lack of specific information about evaders locations.



We chose to analyze the performance of the system with a fan-shaped sensor as this type of sensor is highly common in many sensing and scanning applications, from optical to radar and sonar. Fan-shaped sensors with a variable half-angle resemble actual pinhole camera visual sensors of a given aperture. Furthermore, a fan shaped sensor with a larger area can reduce the critical speed and sweep time compared to approaches that use linear sensors while avoiding the usage of unrealistic sensors such as circular sensors which assume to detect evaders in all angles around a searcher.

The considered protocol may be implemented in a $2D$ environment in which the actual agents travel on a plane or as a $3$ dimensional search where the sweepers are drone-like agents which fly over the evader area at different heights.


\textbf{Guaranteed evader detection in practical robotic applications.} A wide range of real-world tasks that are nowadays carried out by human-controlled machines are expected to be replaced by partially autonomously operated robots in the nearby future. Search and rescue missions, airborne surveillance applications, various monitoring tasks for security applications, wildlife tracking, fire control as well as inspection tasks in hazardous zones can all benefit from the theoretical and experimental results developed in this work. The combination of the proposed search protocol and sensor choice enables nearly optimal cooperation between agents and allows the deployment of multi-robot teams with superior performance. 

For the mentioned applications, guaranteeing success in the worst-case scenario ensures succeeding in the task for all other simpler scenarios as well. This approach is often used in real-world settings where full state information is not available and performance guarantees must be kept. 


The searching agents considered in this work do not assume knowledge of the number of evaders present in the region, their locations, or their escape plan and despite that they are able to detect all of them. Therefore, this work is of prime theoretical and practical importance as is in many pursuit-evasion games the searching team does not have complete information about its opposing team, as is often assumed by many previous papers. 

Since multi-agent pursuit-evasion search protocols mainly utilize multi-agent UAVs, sweepers fly over the environment containing the evaders, therefore investigating issues such as obstacles is not the main focus of the work, because the sweeping team flies over them. Obstacles limit the movements and locations of ground-moving evaders, and therefore their presence assists the searching team to detect them since it limits the escape options of evaders, and thus does not impact our “worst-case” analysis. 

The mentioned protocols can use a vast suite of onboard sensors to detect evaders, depending on the domain of application. Potential choices vary from visual sensors such as cameras which have the benefit of having a high resolution and being lightweight. Therefore, detecting evaders with cameras requires a smaller battery in order to accomplish the desired task compared to other sensing modalities such as radars that increase the weight of the payload and hence limit the duration of the search mission due to increased energy consumption. Actual detection of evaders can utilize a vast number of computer-vision detection algorithms such as \cite{sandler2018mobilenetv2, bochkovskiy2020yolov4}. 

Since to our use case, the preferred choice for the detection modality is a camera, we extend and generalize previous works on guaranteed detection of smart targets to accommodate usage of such sensors. Previous works such as used circular sensors and linear sensors. However, since both circular and linear sensors offer simplistic assumptions about the area detected by actual searching robot teams, in this work we extend and generalize state-of-the art results on guaranteed detection of smart evaders that use pincer sweeps between searching pairs to search teams that use sensors modelling actual visual detectors. Our obtained results are insensitive to locations of evaders or their numbers. The proposed protocols can be applied in other convex environments as well, by using slight modifications to the explored sweeping strategies.

\textbf{Overview of related research.} 
Several interesting search problems originated in the second world war due to the need to design patrol strategies for aircraft aiming to detect ships or submarines in the English channel, see \cite{koopman1980search}. Patrolling a corridor with multi agent teams whose goal is ensuring detection and interception of smart evaders was also investigated in \cite{vincent2004framework} while optimally proven  strategies were provided in \cite{altshuler2008efficient}. A somewhat related, discrete version of the problem, was also investigated in \cite{altshuler2011multi}. It focuses on a dynamic variant of the cooperative cleaners problem, a problem that requires several simple agents to a clean a connected region on a grid with contaminated pixels. This contamination is assumed to spread to neighbors at a given rate.

In \cite{bressan2008blocking,bressan2012optimal}, Bressan et al. investigate optimal strategies for the construction of barriers in real-time aiming at containing and confining the spread of fire from a given initial area of the plane. The authors are interested in determining a minimal possible barrier construction speed that enables the confinement of the fire, and on determining optimality conditions for confinement strategies.

A non-escape search procedure for evaders that are originally located in a convex region of the plane from which they may move out of, is investigated in \cite{tang2006non}, and a cooperative progressing spiral-in algorithm performed by several agents with disk shaped sensors in a leader-follower formation is proposed. In \cite{mcgee2006guaranteed}, McGee et al. investigate guaranteed search patterns for smart evaders that do not have any maneuverability restrictions except for an upper limit on their speed. The sensor the agents are equipped with detects evaders within a disk shaped area around the searcher's location. Search patterns consisting of spiral and linear sections are considered. In \cite{hew2015linear}, Hew investigates search for smart evaders by implementing concentric arc trajectories with agents having disk-shaped sensors similar to the ones used in \cite{mcgee2006guaranteed}.  The aim of search protocol is to detect submarines in a channel or in a half plane.

Another set of related problems are pursuit-evasion games, where the pursuers' objective is to detect evaders and the evaders objective is to avoid the pursuers. Pursuit-evasion games include combinations of single and multiple evaders and pursuers scenarios. These types of problems are addressed in the context of perimeter defense games by Shishika et al. in \cite{shishika2018local,shishika2020cooperative}, with a focus on utilizing cooperation between pursuers to improve the defense strategy. In \cite{shishika2018local}, implicit cooperation between pairs of defenders that move in a ``pincer movement" is performed to intercept intruders before they enter a planar convex region. In \cite{makkapati2019optimal}, pursuit–evasion problems involving multiple pursuers and multiple evaders (MPME) are studied. The original MPME problem is decomposed to a sequence of simpler multiple pursuers single evader (MPSE) problems by classifying if a pursuer is relevant or redundant for each evader and only the relevant pursuers participate in the MPSE pursuit of each evader. Pursuers and evaders are all assumed to be identical, and pursuers follow either a constant bearing or a pure pursuit strategy. The problem is simplified by adopting a dynamic divide and conquer approach, where at every time instant each evader is assigned to a set of pursuers based on the instantaneous positions of all the players. The original MPME problem is decomposed to a sequence of simpler multiple pursuers single evader (MPSE) problems by classifying if a pursuer is relevant or redundant for each evader by using Apollonius circles. Only the relevant pursuers participate in the MPSE pursuit of each evader.


Recent surveys on pursuit evasion problems are \cite{ chung2011search,kumkov2017zero,weintraub2020introduction}. In \cite{chung2011search}, a taxonomy of search problems is presented. The paper highlights algorithms and results arising from different assumptions on searchers, evaders and environments and discusses potential field applications for these approaches. The authors focus on a number of pursuit-evasion games that are directly connected to robotics and not on differential games which are the focus of the other cited surveys. \cite{kumkov2017zero} presents a survey on pursuit problems with $1$ pursuer versus $2$ evaders or $2$ pursuers versus $1$ evader are formulated as a dynamic game and solved with general methods of zero-sum differential games. In \cite{weintraub2020introduction}, the authors present a recent survey on pursuit-evasion differential games and classify the papers according to the numbers of participating players: single-pursuer single-evader (SPSE), MPSE, one- pursuer multiple-evaders (SPME) and MPME. In \cite{garcia2019optimal}, a two-player differential game in which a pursuer aims to capture an evader before it escapes a circular region is investigated. In \cite{garcia2020multiple}, the problem of border defense differential game where $M$ pursuers cooperate in order to optimally catch $N$ evaders before they reach the border of the region and escape is investigated.

\textbf{Contributions:} We present a comprehensive theoretical and numerical analysis of trajectories, critical speeds and search times for a team of $n$ cooperative sweeping agents equipped with fan-shaped sensors with a variable half-angle, whose mission is to guarantee detection of all smart evaders that are initially located inside a given circular region from which they may move out of to escape the pursuing sweeping agents.
\begin{itemize}
        \item  We present a novel pincer search strategy utilizing pincer sweeps and complementing sensor geometries to improve the detection capabilities of the search team.
        \item We develop analytic formulas for a search protocol, applicable to any even number of sweepers with fan-shaped sensors with a given half-angle. Fan shaped sensors with a variable half-angle that model the actual field-of-view of visual sensors, therefore enabling the applicability of the established results in real-world search scenarios.  
        \item We extend state-of-the-art multi-pursuer multi-evader literature to scenarios with arbitrary large numbers of evaders, where very limited information is available to the pursuers, and even so they are able to optimally cooperate in order to successfully complete their mission.   
\end{itemize}

The research performed in this paper extends and generalizes previous results on linear and circular sensors such as \cite{francos2021search, mcgee2006guaranteed} to fan-shaped sensors with an arbitrary angle. Since “at the limit”, a linear shaped sensor is a special case of a fan shaped sensor with an angle of $0$ and circular or disk-shaped sensors are fan shaped sensors with an angle of $2\pi$.

Hence the analysis performed in this work provides a significant theoretical milestone in generalizing previous results and allowing the application of the established results to practical robotic vision-based search tasks.

\textbf{Numerical Evaluation.} The theoretical analysis is complemented by simulation experiments in MATLAB and NetLogo that verify the theoretical results, highlight the performance of search strategies with different number of sweepers and sensors and illustrates them graphically in the figures embedded throughout the text and in the accompanying video.

\textbf{Paper Organization.} The paper is organized as follows: Section $2$ describes the motivation and setting for using pincer-based search strategies. Section $3$ presents a optimal lower bound on the sweepers' speed that is independent of the specific chosen sweep protocol. In section $4$, an analytical analysis of critical speeds and sweep time is provided, accompanied by numerical results. In the last section conclusions are given and future research directions are discussed.

\section{Pincer-Based Search Problem Formulation}

Inherently, the complete search time of the region depends on the sweeping protocol the sweeping team implements. The protocols we propose aim to reduce the radius of the circle bounding the evader region after each full sweep. This guarantees complete elimination of the potential area where evaders may be located. 

At the beginning of the search, we assume the entire area of the sweepers' sensors is inside the evader region. This implies that the full length of the central line of the sweepers' sensors (see the light green lines that depicts this part of the sensor in Fig. $1$) is inside the evader region, i.e., a footprint of length $2r$. The sweepers' sensors are shown in green in Fig. $1$ as well. The blue circle on the light green line depicts the center of the sensor of clockwise sweeping agent (and therefore its center of mass). $\alpha$ denotes the half-angle of the sweepers fan-shaped sensors.

If we were to distribute sweepers equally along the boundary of the initial evader region, and have them move in the same direction, potential escape from the points adjacent to the starting locations of the sweepers might occur. To enable sweepers to succeed in the task while having the lowest possible critical speed, we propose that pairs of sweepers move out in opposite directions along the boundary of the evader region and sweep in a pincer movement instead of implementing a protocol where all sweepers move in the same direction along the boundary.

The proposed strategy can be applied with any even number of sweepers. Each sweeper is responsible for an angular sector of the evader region that is proportional to the number of sweepers. The sweepers' field-of-views are initially positioned in pairs back-to-back. One sweeper in the pair moves counter-clockwise while the other sweeper in the pair moves clockwise. Once the sweepers meet, i.e. their sensors are again superimposed at a meeting point, they switch the directions in which they move. Changing of directions takes place every time sweepers ``bump" into each other. 


It is worth emphasizing that once a sweeper leaves a location that was cleared from evaders, other evaders may attempt to reach this location again. Therefore, the proposed sweep protocol must ensure that there is no evaders strategy that enables any evader to escape even if evaders wait at the edge of a cleared location and start their escape instantly after a sweeper leaves this location.

Sweepers implementing pincer movements solve the problem of evader region's spread from the "most dangerous points". Evaders located at these points have the maximum time to escape throughout the sweepers' movements. Consequently, if evaders attempting to escape from these locations are detected, evaders attempting to escape from all other locations are detected as well. When a pair of sweepers travelling in a pincer movement finishes sweeping its allocated section of the environment (particularly an angle of $\frac{2 \pi}{n} - \gamma$, which is explained in detail in the next section), provided they move at a speed exceeding the critical speed, the spread of the evader region that originates from these points has to be less than $2r$. Therefore, the sweepers advance towards the center of the evader region by the margin between the spread of the evader region during this motion and $2r$.


Since a pair of sweepers begin their sweep when the footprints of their sensors are “back-to-back”, the evader region’s points that should be considered to guarantee detection of all evaders are located at the inner tips of the central part of the sweepers’ sensors closest to the center of the evader region and not from points on the boundary of the evader region. 

If all sweepers move in the same direction following their equally spaced placement around the region, the evader region’s points that need to be considered for limiting the region’s spread are points located on the boundary of the evader region. This will in turn result in higher critical speeds for sweepers implementing same-direction sweeps. Requiring higher critical speeds also results in longer sweep times for same-direction protocols compared to pincer based methods. 

In \cite{francos2021pincer}, a quantitative analysis comparing critical speeds and sweep times between pincer-based search protocols and same-direction protocols operated by the same teams of sweepers is investigated for agents with different sensors and sweep protocols then the ones considered in this work. The analysis in \cite{francos2021pincer} indicates that the critical speeds and sweep times for agents employing same direction sweeps is indeed higher compared to the pincer-based speeds. We assume this applies in this settings as well and plan to analytically compare these search methodologies in future work.


\begin{figure}[ht]
\noindent \centering{}\includegraphics[width=2.1in,height =1.7808in]{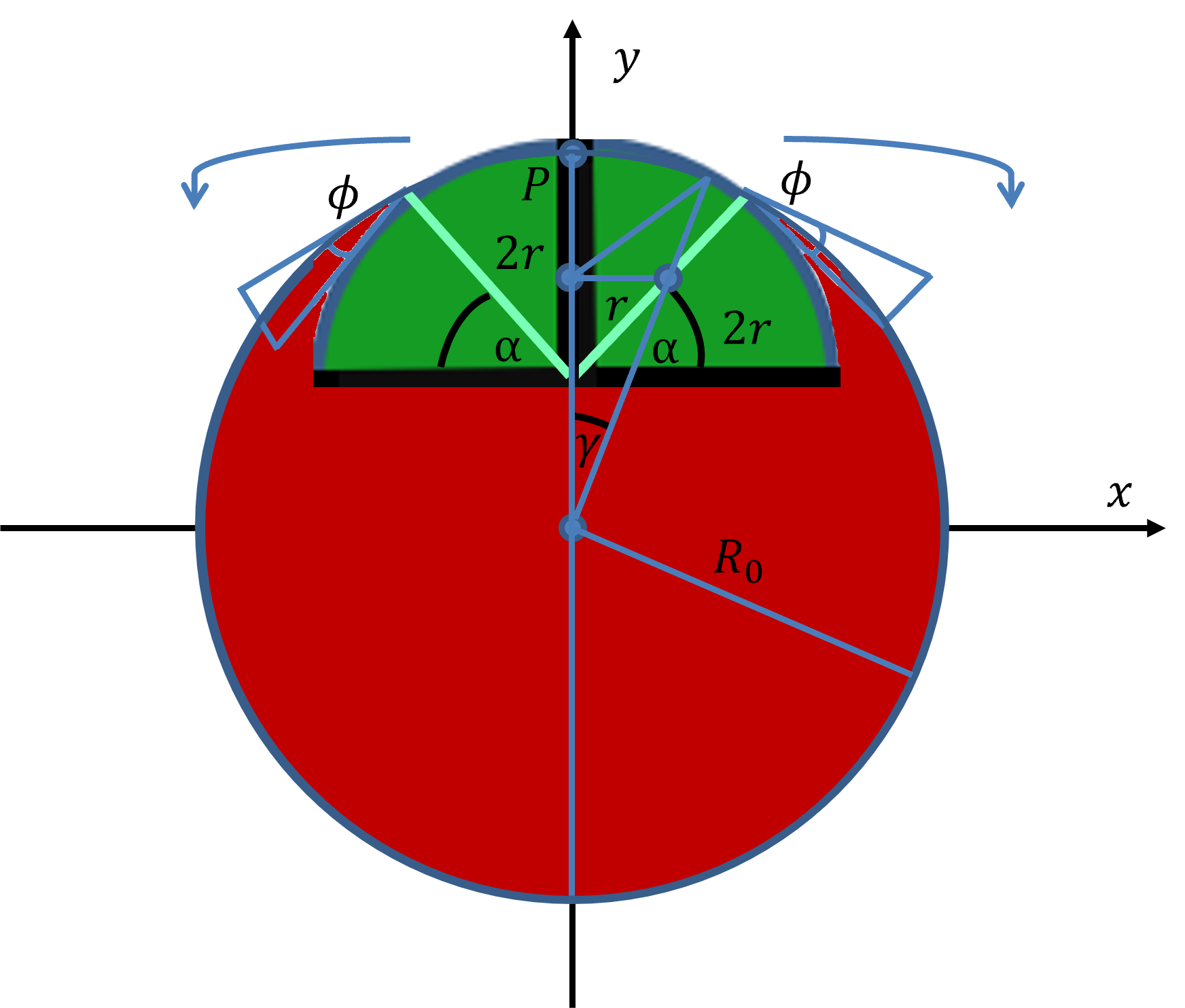} \caption{Initial placement of $2$ sweepers employing the spiral pincer sweep strategy.}
\label{Fig1Label}
\end{figure}

The proposed search pattern uses spiral scans, related to the sweep pattern suggested in \cite{mcgee2006guaranteed}. In order to have maximal sensor footprint intersecting the evader region throughout the sweep, the proposed search pattern aims to track the expanding evader region's ``wavefront".

Simulations demonstrating the evolution of the search strategies are generated using NetLogo software \cite{tisue2004netlogo} and shown in Fig. $2$. Green areas are locations that are free from evaders and red areas indicate locations where potential evaders may still be located. Fig. $2$ shows the cleaning progress of the evader region when $4$ sweepers employ the proposed sweep protocol. 

\begin{figure}[ht]
\noindent \centering{}\includegraphics[width=2.28in,height =2.3in]{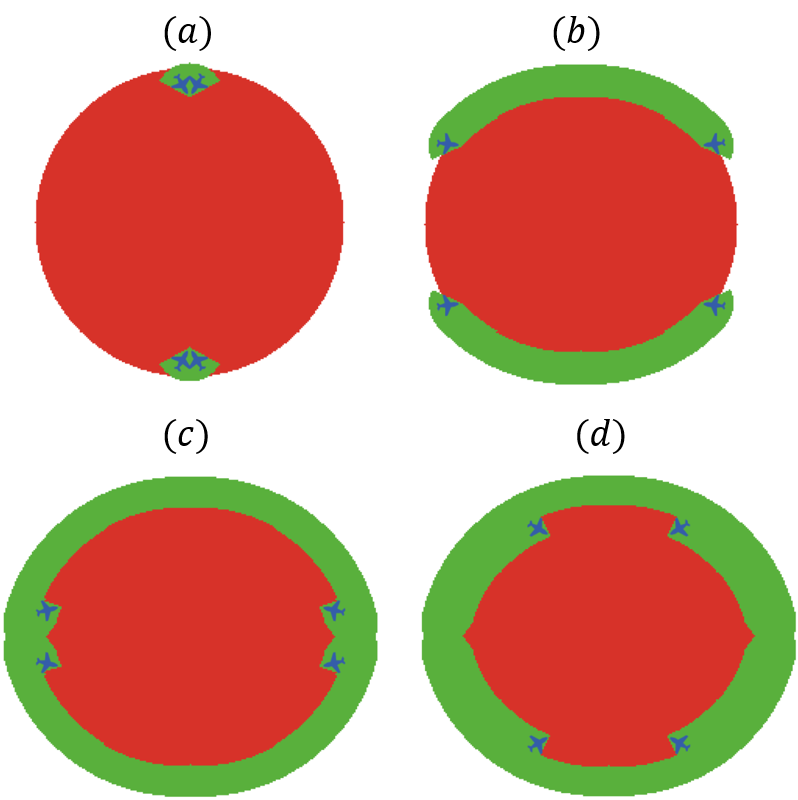} \caption{Swept areas and evader region status for different times in a scenario where $4$ agents employ the spiral pincer sweep process and $\alpha=30^\circ$. (a)-shows the status at beginning of first cycle, (b)- shows the status midway through the first cycle, (c)-shows the status at beginning of the second cycle and (d) shows the status toward the end of the second cycle.}
\label{Fig2Label}
\end{figure}

\textbf{Comparison to related research.} In \cite{francos2019search}, the confinement and complete detection tasks for a line formation of agents or alternatively for a single agent with a linear sensor are analyzed. The presented approach ensures detection of all evaders, however the complete detection time and the critical speeds required for the sweeping formation to succeed in its task are inferior to the results achieved in this work, since no pincer sweeps are performed and since the line formation performs simple circular motion and does not track the expanding wavefront of potential evader locations as is performed in this work.  In \cite{francos2021search}, teams of agents perform pincer sweep search strategies with linear sensors. These two early works assume linear sensors and not fan-shaped sensors as the ones assumed in this paper, hence this paper can be regarded as a generalization of the results from \cite{francos2021search} to fan-shaped sensors with a given half-angle that model the field-of-view of actual cameras allowing applicability of the results to real robotic search and surveillance missions. 

While automated discovery of search policies relies on knowledge of the players’ locations throughout the search process (as assumed in \cite{shishika2018local,shishika2020cooperative}), to plan the trajectories of pursuers and evaders, in our setting we bypass the need for such extensive and often unrealistic knowledge of the environment and perfect communication exchange between members of the same team with a simple and efficient strategy. Furthermore, in contrast to \cite{shishika2018local,shishika2020cooperative} we do not assume a finite number of evaders and that each sweeper can intercept only a single evader, and guarantee interception of all evaders.

As opposed to our work, references \cite{mcgee2006guaranteed, hew2015linear} use a disk shaped sensor with a radius of $r$, and do not calculate the detection times of all evaders. Although the works in references \cite{shishika2018local,shishika2020cooperative} use pincer movements between pairs of defenders as well, they have a different objective of protecting an initial region from invaders, contrary to our goal which is to detect all evaders that may spread from the interior of the region. Furthermore these works are concerned with devising policies for intercepting as many intruders as possible relying on an assumption that the number of intruders is finite and that each defender can only intercept a single intruder and not on guaranteeing interception of all evaders or intruders, as is the aim of our work.

Contrary to the work reported in references \cite{garcia2019optimal, garcia2020multiple}, we do not assume the number of evaders is known to the pursuers and not that the members of each party have full access to the locations of the members of the opposing party and use it in order to plan their actions. From our point of view using such information is unrealistic for search applications where determining the locations of the opposing team members is at the heart of the problem. Such knowledge requires very sophisticated sensing capabilities when searching large regions. Furthermore, we solve the assignment of pursuers to evaders elegantly by assigning sweeper pairs based on their geometric location. This simple assignment rule alleviates the need for communicating synchronized and precise location information between the searchers.

\section{Lower Bound on the Speed of Sweepers with Fan-Shaped Sensors}

\noindent We present an optimal bound on the speed of sweepers equipped with fan shaped sensors. The developed bound is irrespective of the choice of the implemented sweeping protocol. This bound serves as one of the benchmarks that are used to compare the performance of different search strategies by a metric we refer to as the critical speed.  

The largest number of evaders are detected when the entire fan shaped sensor overlaps the evader region. For a fan shaped sensor of length $2r$ and half angle $\alpha$ the maximal rate of detecting evaders must be higher than the minimal expansion rate of the evader region. Else, there is no feasible sweeping protocol that can ensure the detection of all evaders. 

The lower bound is derived for a sweeper team containing $n$ identical sweepers. The smallest sweeper's speed that ensures the maximal detection rate is larger than the minimal expansion rate is based purely on geometric properties of the evader region, the sweepers sensors' dimensions and the evaders' maximal speed.  This lower bound on the speed is defined as the critical speed and is denoted by $V_{LB}$.

\begin{equation}
V_{LB} = \frac{{\pi {R_0}{V_T}}}{nr}
\label{e1}
\end{equation}

This lower bound is similar for line sensors and its proof follows the derivation such a lower bound for a sweeper's critical speed in \cite{francos2021search}. Similar results for the bound in \cite{francos2021search} for sweepers with circular sensors are given in \cite{mcgee2006guaranteed}.

If sweepers move at speeds that exceed the critical speed, they possess the ability to implement a suitable sweep protocol that decreases the evader region, therefore allowing them to sweep the next iteration around a smaller region. Therefore, the critical speed is calculated based on the sweep around the initial and largest radius, since sweeping around smaller regions takes less time. Therefore, if the sweepers succeeded in confining evaders to the larger and initial evader region they will surely succeed when the evader region decreases throughout the search protocol (since their speed stays constant throughout the protocol).

\section{Spiral Pincer Sweep Strategies For Sweepers with Fan-Shaped Sensors}

\subsection{The Critical Speed}
In order for the sweeper team to efficiently scan the region, we desire that there will be maximal intersection between the sweepers' sensors and the evader region. Implementing spiral trajectories, in which sweepers' sensors track the expanding evader region's wavefront, allows the sweeping team to achieve nearly optimal efficiency as proven in this section. 

We consider a scenario in which a sweeping team consisting of $n$ sweepers, where $n$ is even. All sweepers have fan-shaped sensors with finite visibility that model the sensor geometry of real cameras. Sweepers' sensors have a diameter length of $2r$ and a half-angle of $\alpha$, as shown in Fig. $1$. Selecting sweepers with complementing sensor geometries along with having the sweepers move in opposite directions prevents evaders from devising an escape strategy that leverages the gap between the sweepers' sensors, since the field-of-views of each sweeping pair are tangent to each other. These considerations highlight the benefits of using both pincer sweeps along with sensor geometries between sweeping pairs that complete each other in order to achieve maximal evader detection performance.

Fig. $1$ shows the initial setting of the problem when $2$ sweepers perform the search. The symmetry between the trajectories of adjacent searching pairs prevents potential escape of evaders from point $P=(0,R_0)$,  "the most dangerous point`` evaders may attempt to escape from due to similar consideration as are proven in \cite{francos2019search}.

Hence, a sweeper's critical speed depends solely on the time required for it to sweep its allocated angular sector. Define by $\gamma$ the angle describing the offset angle of the center of the sweeper from the center of the region at the beginning of the sweep process, see Fig. $1$ for a geometric illustration of $\gamma$. The relation between $\gamma$ and $\alpha$ obtained by using the law of cosines on the depicted triangles of Fig. $1$ and is given by,
\begin{equation}
4{r^2} = {R_0}^2 + {\left( {{R_0} - 2r + r\cos \alpha } \right)^2} - 2{R_0}\left( {{R_0} - 2r + r\cos \alpha } \right)\cos \gamma 
\label{e45}    
\end{equation}
Rearranging terms yields,
\begin{equation}
\begin{array}{l}
\gamma  = \arccos{ \left( {\frac{{2{R_0}^2 + 2{R_0}r\left( {\cos \alpha  - 2} \right) + {r^2}\cos \alpha \left( {\cos \alpha  - 2} \right)}}{{2{R_0}\left( {{R_0} + r\cos \alpha  - 2r} \right)}}} \right)}
\label{e300}
\end{array}
\end{equation}
Therefore, at each sweep agents traverse an angle of $\frac{2\pi}{n} - \gamma$.  After sweepers complete searching their allocated sector, and only if they move at a sufficient speed, they advance toward the center of the evader region together. If the search is planar, sweepers switch the sweeping directions following an inward movement toward the center. If the search is $3$ dimensional, sweepers first move together toward the center of the region and only afterwards exchange the sector they sweep with the sweeping teammate they perform the pincer sweep with. Following this motion they commence sweeping an evader region bounded by a circle with a smaller radius.

Sweepers begin their spiral motion with the outer tip of the central line of their fan-shaped sensor tangent to the evader region's boundary. To preserve the outer central tip of their sensor tangent to the evader region, sweepers move with an angular offset of $\phi$ to the normal of the evader region (at each point) throughout their sweep. $\phi$ is the angle between the outer tip of the central line of a sweeper's sensor and the normal of the evader region at the meeting point between the evader region and the tip of the central line of the sensor that is furthest from the center of the evader region (see the light green line that depicts this part of the sensor in Fig. $1$). $\phi$ depends on the ratio between the sweeper and evader speeds (see Fig. $1$ for its depiction). The sweepers' incentive to move at a constant angle $\phi$ to the normal of the evader region is to preserve the evader region's circular shape and to keep as much of their sensor footprint inside the evader region at all times in order to detect a maximal number of evaders. $\phi$ is given by,
\begin{equation}
\sin \phi  = \frac{{{V_T}}}{{{V_s}}}
\label{e1000}
\end{equation}

Because spiral sweeping preserves the evader region's circular shape, as a result of the isoperimeteric inequality such trajectories necessitate that the length of the curve bounding the evader region is minimal and therefore the time required for the sweepers to sweep around it is minimal as well. Denote by $\theta _s$ the sweeper's angular speed, given by,
\begin{equation}
\frac{{d{\theta _s}}}{{dt}} = \frac{{{V_s}\cos \phi }}{{{R_s}(t)}} = \frac{{\sqrt {{V_s}^2 - {V_T}^2} }}{{{R_s}(t)}}
\label{e1002}
\end{equation}
Integration of (\ref{e1002}) with the initial and final sweep times of the angular sector as the integral's limits is,
\begin{equation}
\int_0^{{t_\theta }} {\dot \theta \left( \zeta  \right)} d\zeta  = \int_0^{{t_\theta }} {\frac{{\sqrt {{V_s}^2 - {V_T}^2} }}{{{V_T}\zeta  + {R_0} - r}}d} \zeta 
\label{e36}
\end{equation}
With a solution for $\theta \left( {{t_\theta }} \right)$ obtained from (\ref{e36}) given by,
\begin{equation}
\theta \left( {{t_\theta }} \right) = \frac{{\sqrt {{V_s}^2 - {V_T}^2} }}{{{V_T}}}\ln \left( {\frac{{{V_T}{t_\theta } + {R_0} - r}}{{{R_0} - r}}} \right)  
\label{e15}
\end{equation}
Raising by an exponent (\ref{e15}) yields,
\begin{equation}
\left( {{R_0} - r} \right){\exp\left({\frac{{{V_T}{\theta({t_\theta }) }}}{{\sqrt {{V_s}^2 - {V_T}^2} }}}\right)} = {V_T}{t_\theta } + {R_0} - r = {R_s}({t_\theta })
\label{e38}
\end{equation}

The time required for a sweeper to scan its allocated angular section corresponds to it sweeping an angle of $\theta$ by $\frac{2\pi}{n}- \gamma$ around the region. While the sweeper performs this motion, the evader region's expansion must be at most $2r$ from its initial radius.  Otherwise, the sweepers will not be to detect potential evaders attempting to escape from the region. Allowing a spread of at most $2r$ is a simplification of the problem that assumes that while the sweepers progress toward the center of the evader region the evaders do not spread. Obviously this is not a realistic assumption, hence we address this case after the solution of the simplified setting. To guarantee that no potential evader escapes the sweepers, following a sweep by $\frac{2\pi}{n}-\gamma$ we must demand that,
\begin{equation}
{R_0} + r \ge {R_s}({t_{\frac{2\pi}{n}-\gamma}})
\label{e39}
\end{equation}
Define,
\begin{equation}
\lambda  \buildrel \Delta \over = \exp\left({\frac{{(\frac{2\pi}{n}-\gamma) {V_T}}}{{\sqrt {{V_s}^2 - {V_T}^2} }}}\right)
\label{e40}    
\end{equation}
Replacing ${R_s}({t_{\frac{2\pi}{n}-\gamma}})$ with the expression derived for the trajectory of the sweeper's center results in, ${R_0} + r \ge \left( {{R_0} - r} \right)\lambda$. Hence, to guarantee no potential evader escapes without being detected by the sweeper team, it is mandatory for the sweepers' speed to exceed,
\begin{equation}
{V_S} \ge {V_T}\sqrt {\frac{{{{\left( {\frac{{2\pi }}{n}-\gamma} \right)}^2}}}{{{{\left( {\ln \left( {\frac{{{R_0} + r}}{{{R_0} - r}}} \right)} \right)}^2}}} + 1}
\label{e43}
\end{equation}

As mentioned earlier, in order to accommodate the expansion of evaders during the inward motion of sweepers a modification to the critical speed in (\ref{e43}) must be made. This change requires that when sweepers move inwards after they finish their sweep, they must meet the evader region's expanding wavefront that moves outwards from every point in the evader region with a speed of $V_T$ at the previous radius $R_0$. Incorporating this constraint into the solution of the critical speed guarantees that no evaders escape during the sweepers inwards motion as well. Denote by ${T_c}$ the evader region's expansion throughout the first sweep. In order for no evader to escape detection the following inequality must hold,  ${V_T}{T_c} \leq \frac{{2r{V_s}}}{{{V_s} + {V_T}}}$. Replacing ${T_c}$ with $\frac{\left( {{R_0} - r} \right)\left( {\lambda - 1} \right)}{V_T}$ results in,

\begin{equation}
\left( {{R_0} - r} \right)\left( {\lambda - 1} \right) = \frac{{2r{V_s}}}{{{V_s} + {V_T}}}
\label{e49}
\end{equation}
\newtheorem{thm}{Theorem}
\begin{thm}
For a spiral pincer sweep protocol with $n$ sweepers in which $n$ is even, the critical speed , ${V_s}$, that allows the sweeping team to confine all evaders to their original domain is computed from,
\begin{equation}
{V_s} = {V_T}\frac{1}{{1 - \frac{{2r}}{{\left( {{R_0} - r} \right)\left( {\lambda  - 1} \right)}}}}
\label{e98996}
\end{equation}
\end{thm}

The solution for the critical speed is calculated numerically by applying the Newton–Raphson method. The simplified critical speed of (\ref{e43}) serves as an initial guess. In the rest of the article the critical speed considered is that of theorem $3$, thus it accounts for the expansion of the evader region during the sweepers advancements into the center of the region. 




\subsection{Analytical Sweep Time Analysis}

\begin{thm}
For a team of $n$ sweepers implementing the pincer sweep protocol, the time required to detect all evaders in the region and reduce the evader region's area to $0$ is given by the sum of inward advancement times after all the sweeps,${T_{in}}(n)$, along with the sum of all the spiral traversal times, ${T_{spiral}}(n)$ in all sweeps. Hence,
\begin{equation}
T(n) = {T_{in}}(n) + {T_{spiral}}(n)
\label{e98993}
\end{equation}
\end{thm}

\begin{proof}

Denote by $\Delta V > 0$ the excess speed of a sweeper above the critical speed. The sweeper's speed hence is $V_s = V_c + \Delta V$. Let $\theta \left( {{t_\theta }} \right)$ denote the angle of a sweeper with respect to the center of evader region. At the start of each sweep the center of a sweeper's sensor is located at a distance of ${R_i} - r$ from the center of the evader region. $\theta \left( {{t_\theta }} \right)$ is calculated in (\ref{e15}). Replacing $R_0$ with $R_i$ yields,
\begin{equation}
\theta \left( {{t_\theta }} \right) = \frac{{\sqrt {{V_s}^2 - {V_T}^2} }}{{{V_T}}}\ln \left( {\frac{{{V_T}{t_\theta } + {R_i} - r}}{{{R_i} - r}}} \right)
\label{e46}
\end{equation}
Denote by ${T_{spiral}}_i$ the time required for a sweeper to sweep an angle of $\theta \left( {{t_\theta }} \right)=\frac{2\pi}{n} -\gamma_i$. This time can be calculated by rearranging terms in (\ref{e46}) and equals, 
\begin{equation}
{T_{spiral}}_i = \frac{{\left( {{R_i} - r} \right)\left( {\lambda - 1} \right)}}{{{V_T}}}
\label{e47}
\end{equation}
In case sweepers move at speeds exceeding the critical speed appropriate for the scenario, denote the distance every sweeper advances toward the center of the evader region by ${\delta _i}(\Delta V)$. Following this inward motion, the evader region decreases and is contained inside a smaller circular evader region having a radius of ${R_{i + 1}} = {R_i} - {\delta _i}(\Delta V)$. Hence ${\delta _i}(\Delta V)$ equals,
\begin{equation}
{\delta _i}(\Delta V) = 2r - {V_T}{T_{spiral}}_i \hspace{1mm}, \hspace{1mm} 0  \le {\delta _i}(\Delta V) \le 2r
\label{e1015}
\end{equation}
The number of sweepers, the half-angle of their sensors and the sweep cycle number (the number representing the number of sweeps already completed by the sweeping team around the evader region) all influence the distance sweepers are able to progress inwards toward the center of the evader region after completing a sweep. If the evader region was not expanding throughout the sweepers' inward motion, then  ${\delta _i}(\Delta V)$ is,
\begin{equation}
{\delta _i}(\Delta V) = 2r - \left( {{R_i} - r} \right)\left( {\lambda - 1} \right)
\label{e1067}
\end{equation}
In the expression ${\delta _i}(\Delta V)$, $\Delta V$ expresses the sweeper's excess speed above the critical speed. Denote by $i$ the sweep cycle number that starts from sweep number $0$. As mentioned earlier in the development of the critical speed, the time required for sweepers to advance toward the center of the evader region up to the point in which their sensors fully overlap the evader region depends on the relative speed between the sweepers inwards entry speed and the evader region outwards expansion speed. Hence, after sweepers finish a sweep they advance toward the center of evader region by a distance of,
\begin{equation}
{\delta _{{i_{eff}}}}(\Delta V) = {\delta _i}(\Delta V)\left( {\frac{{{V_s}}}{{{V_s} + {V_T}}}} \right)
\label{e700}
\end{equation}

Following the completion of a sweep around the region, sweeper pairs progress together in the direction of the evader region's center. During inwards motions, sweepers move at speed of $V_s$, up to the point in which they begin to spirally sweep again once their sensors are fully over the evader region’s expanding wavefront. Therefore, their speed is always bounded. We assume as a worst-case assumption, that the sweeper's sensors detect evaders only when sweepers perform spiral motions. Therefore, throughout inward motions no detection of evaders occurs, while the evader region continues to expand due to motions of evaders. In the video accompanying the paper, the time it takes the sweepers to advance inwards is taken into account and dictates the new radius of the evader region after the sweep. Hence, after an inward advancement the evader region is within a circle whose radius is,

\begin{equation}
{R_{i + 1}} = {R_i} - {\delta _i}(\Delta V)\left( {\frac{{{V_s}}}{{{V_s} + {V_T}}}} \right)
\label{e1097}
\end{equation}
Denote by $\widetilde {R}_i=R_i -r$. The usage of $\widetilde {R}_i$ allows to analytically solve for the number of sweeps required to complete the search of the entire evader region. Replacing ${{\delta _i}(\Delta V)}$ into (\ref{e1097}) yields the following difference equation,
\begin{equation}
\widetilde{R}_{i + 1} = \widetilde{R}_i\left( {\frac{{{V_T} + {V_s}\lambda}}{{{V_s} + {V_T}}}} \right) - \frac{{2r{V_s}}}{{{V_s} + {V_T}}}
\label{e1073}
\end{equation}
Denote the coefficients of (\ref{e1073}) by,
${c_1} =  - \frac{{2r{V_s}}}{{{V_s} + {V_T}}}, {c_2} = \frac{{{V_T} + {V_s}\lambda}}{{{V_s} + {V_T}}}$. Hence,
\begin{equation}
\widetilde{R}_{i + 1} = {c_2}\widetilde{R}_i +{c_1}
\label{e345}    
\end{equation}
Denote by ${R_N}$ the radius of the circle bounding the evader region when it is shrunk to be within a circle having a radius smaller or equal to $2r$. Calculating ${R_N}$ is possible only after the number of sweeps around the region, $N_n$, is calculated. Hence, we use $\widehat{R}_N=2r$ as an estimate of ${R_N}$ to allow the calculation of $N_n$. The number of sweeps required to decrease the evader region to be within a circle of radius $\widehat{R}_N=2r$ is,
\begin{equation}
{N_n} = \left\lceil {\frac{{\ln \left( {\frac{{r\left( {3 - \lambda} \right)}}{{{R_0}\left( {1 - \lambda} \right) + r\left( {1 + \lambda} \right)}}} \right)}}{{\ln \left( {\frac{{{V_T} + {V_s}\lambda}}{{{V_s} + {V_T}}}} \right)}}} \right\rceil
\label{e1076}
\end{equation}

The ceiling operator is used since the number of sweeps must be integer. This means that sweepers continue sweep number $N_n$ even if the evader region's radius is reduced to less $2r$ throughout the final spiral sweep. Hence, to facilitate the calculation of $N_n$ we assume the last sweep occurs when the evader region is within a circle of radius $\widehat{R}_N = 2r$. Denote by $T_{i{n_i}}$ the duration of each inwards motion. It is a function of the sweep cycle number and equals,

\begin{equation}
{T_{i{n_i}}} = \frac{{{\delta _{{i_{eff}}}}(\Delta V)}}{{{V_s}}} = \frac{{2r - {\widetilde{R}_i}\left( {\lambda - 1} \right)}}{{{V_s} + {V_T}}}
\label{e1077}
\end{equation}

Denote the sum of all inward advancement times up to the point in which the evader region is within a circle of radius smaller than or equal to $2r$ as $\widetilde{T}_{in}(n)$, i.e. $\widetilde{T}_{in}(n) = \sum\limits_{i = 0}^{{N_n}-2} {{T_{i{n_i}}}}$. During inwards motions the sweepers' sensors are not fully inside the evader region, hence they detect evaders in locations that are already free from evaders. Hence, we assume that sweepers do not detect evaders until they complete their inward motion and begin sweeping again. The total sweep times until the evader region is within a circle having a radius less than or equal to $2r$ is computed as the addition of the total spiral sweep times and the total inward advancement times. Therefore,
\begin{equation}
\widetilde{T}(n) = \widetilde{T}_{in}(n) + \widetilde{T}_{spiral}(n)
\label{e1083}
\end{equation}
The sum of all inward motion times until the evader region is reduced to be within a circle having a radius less than or equal to $2r$ is,
\begin{equation}
\begin{array}{l}
\widetilde{T}_{in}(n) = \sum\limits_{i = 0}^{{N_n} - 2} {{T_{i{n_i}}}}  = \frac{{2r}}{{{V_s} + {V_T}}} + \frac{{{R_0} - r}}{{{V_s}}} + \frac{{2r\left( {{V_T} + {V_s}\lambda} \right)}}{{{V_s}\left( {{V_s} + {V_T}} \right)\left( {1 - \lambda} \right)}}\\
 - \frac{{{{\left( {{V_T} + {V_s}\lambda} \right)}^{{N_n} - 1}}}}{{{V_s}\left( {{V_s} + {V_T}} \right)\left( {1 - \lambda} \right)}}\left( {{R_0}\left( {1 - \lambda} \right) + r\left( {1 + \lambda} \right)} \right)
\end{array}
\label{e702}
\end{equation}

During the final inward motion, sweepers advance toward the region's center and place the inner tips of the central part of their sensors at the center of the evader region. Afterwards, the sweepers perform a final circular sweep and complete the detection of all evaders in the region. The time required to perform this motion is denoted by ${T_{_{in}last}}(n)$ and is given by ${T_{_{in}last}}(n) = \frac{{{R_{{N}}} }}{{{V_s}}}$, resulting in,
\begin{equation}
{R_N} =  - \frac{{2r}}{{1 - \lambda}}
 + {c_2}^{{N_n}}\left( {\frac{{{R_0}\left( {1 - \lambda} \right) + r\left( {1 + \lambda} \right)}}{{1 - \lambda}}} \right)
\label{e782}
\end{equation}
By replacing the exact value of ${R_{{N}}}$ from (\ref{e782}) in ${T_{_{in}last}}(n)$ we obtain,
\begin{equation}
{T_{_{in}last}}(n) = - \frac{{2r}}{{{V_s}\left( {1 - \lambda} \right)}}
 + {c_2}^{{N_n}}\left( {\frac{{{R_0}\left( {1 - \lambda} \right) + r\left( {1 + \lambda} \right)}}{{{V_s}\left( {1 - \lambda} \right)}}} \right)
\label{e727}
\end{equation}
The time required for sweeping around radius $\widetilde{R}_i$ is computed by multiplying $\widetilde{R}_i$ by $\frac{{\lambda - 1}}{V_T}$. Hence, multiplying (\ref{e1073}) by $\frac{{\lambda - 1}}{V_T}$ yields the following sweep times difference equation,
\begin{equation}
{T_{i + 1}} = {c_2}{T_i} + {c_3}
\label{e37}
\end{equation}
Where ${c_3} = \frac{{ - 2r{V_s}\left( {\lambda - 1} \right)}}{{({V_s}+{V_T}){V_T}}}$. The equation for computing the total spiral sweep times up to the point at which the evader region is within a circle of radius less or equal to $2r$ is,
\begin{equation}
\widetilde{T}_{spiral}(n) = \frac{{{T_0} - {c_2}{T_{N_n - 1}} + \left( {{N_n} - 1} \right){c_3}}}{{1 - {c_2}}}
\label{e100}
\end{equation}
Replacing the coefficients into (\ref{e100}) provides,
\begin{equation}
\begin{array}{l}
\widetilde{T}_{spiral}(n) = \frac{{\left( {r - {R_0}} \right)\left( {{V_s} + {V_T}} \right)}}{{{V_T}{V_s}}} - \frac{{2r\left( {{V_T} + {V_s}\lambda} \right)}}{{{V_T}{V_s}\left( {1 - \lambda} \right)}} + \frac{{2r\left( {{N_n} - 1} \right)}}{{{V_T}}} \\ - {\left( {\frac{{{V_T} + {V_s}\lambda}}{{{V_s} + {V_T}}}} \right)^{{N_n}}}\left( {\frac{{\left( {{V_s} + {V_T}} \right)\left( {{R_0}\left( {\lambda - 1} \right) - r\left( {\lambda + 1} \right)} \right)}}{{{V_T}{V_s}\left( {1 - \lambda} \right)}}} \right)
\end{array}
\label{e1086}
\end{equation}

\subsection{The End-game}

Following the sweepers' completion of sweep number ${{N_n} - 1}$, they progress in the direction of the center of the evader region up to a position in which the inner tips of the central parts of their sensors are positioned at the center of the evader region. After this inward motion, the sweepers must complete a set of maneuvers referred to as the end-game in order to ensure detection of all smart evaders. A final circular sweep around radius $r$ is required to complete the detection of all evaders and complete the search mission. Sweepers have the ability to successfully detect all evaders during the last circular sweep only if their speed is sufficiently high to ensure that although they do not track the expanding wavefront of the remaining evader region at this last sweep with a spiral trajectory and perform a considerably less efficient circular sweep, they are still able to catch all evaders. The reason why this last motion cannot be a spiral out motion is because we wish that the tips of the central part of the sweepers sensors will always be positioned at the center of the evader region to ensure that no evaders remain at the center of the region or near its vicinity. Due to the fact that the critical speed for a spiral sweep is lower compared to a critical speed of circular trajectories, the sweepers can only implement the last circular sweep after spiral sweep number ${{N_n} - 1}$ if their speeds are sufficiently high and satisfy the following inequality,
\begin{equation}
2r \ge {V_T}{T_{last}} + {V_T}{T_{_{in}last}} + {R_{N}}
\label{e703}
\end{equation}
Satisfying (\ref{e703}) implies that detection of all evaders is guaranteed. Prior to the last sweep the evader region is within a circle of radius ${R_{N}}$ such that $0 < {R_{N}} \le 2r$. ${R_{{N}}}$ can be also expressed as, ${R_{{N}}} = r\left( {2 - \varepsilon } \right)$. Hence, $\varepsilon$ may be expressed as $\varepsilon  = \frac{{2r - {R_N}}}{r} , \hspace{1mm} 0 \le \varepsilon  < 2$. The final circular sweep takes place once the sweepers progress to the region's center and place the lower tips of the central part of their sensors at the center of the evader region. Hence, the last circular sweep is spans an angle of $\frac{2\pi}{n}-\gamma_i$ around a region contained within a circle of radius $r$ around the center of the evader region. The time required for sweepers to perform this motion is, 
\begin{equation}
{T_{last}}(n) = \left( {\frac{{2\pi }}{n} - 2\alpha } \right)\frac{r}{{{V_s}}}
\label{e619}
\end{equation}
Denote the smallest possible $\varepsilon$ that satisfies (\ref{e703}) as ${\varepsilon _c}$. To perform the last circular sweep immediately following spiral sweep number ${{N_n} - 1}$, the inequality in (\ref{e703}) implies that $\varepsilon  \ge {\varepsilon _c} = \frac{{2{V_T}\left( {\pi  - n\left( {\alpha  - 1} \right)} \right)}}{{n\left( {{V_T} + {V_s}} \right)}}$. This consideration dictates that implementation of a circular sweep immediately after spiral sweep number ${{N_n} - 1}$, $V_s$ is possible only if,
\begin{equation}
{V_s} \ge \frac{{2{V_T}\left( {\pi  - n\left( {\alpha  - 1} \right)} \right) - n\varepsilon_c {V_T}}}{{n\varepsilon_c }}
\label{e709}
\end{equation}
If ${R_N} \ge r\left( {2 - {\varepsilon _c}} \right)$ or equivalently,
\begin{equation}
{R_N} \ge \frac{{2rn\left( {{V_T} + {V_s}} \right) - 2{V_T}r\left( {\pi  - n\left( {\alpha  - 1} \right)} \right)}}{{n\left( {{V_T} + {V_s}} \right)}}
\label{e713}
\end{equation}
Then the sweepers' speed does not suffice to rule out the possibility of a feasible escape trajectory for potential evaders. Hence, if (\ref{e713}) holds or (\ref{e709}) does not, then the sweepers must implement a final spiral sweep, that starts with the lower tips of the central parts of their sensors at the center of the evader region. This spiral sweep starts when the center of each sweeper is at a distance of $r$ from the center of the region. Denote by ${T_l}(n)$ the time required to sweep it. Hence, ${T_l}(n)$ equals ${T_l}(n) = \frac{{r\left( {\lambda - 1} \right)}}{{{V_T}}}$. Denote by $\eta$ a characteristic function that assumes only two possible values, $1$ or $0$. If the additional spiral sweep has to be implemented than $\eta =1$ and consequently ${T_l}(n)$ is added to the total sweep time, otherwise in case no additional spiral sweep is required $\eta =0$. Hence,
\begin{equation}
{T_l}(n) = \frac{{r\left( {\lambda - 1} \right)}}{{{V_T}}}
\label{e715}
\end{equation}

\begin{equation}
{T_{spiral}}(n) = \widetilde{T}_{spiral}(n) + {T_{last}}(n) + \eta {T_l}(n)
\label{e719}
\end{equation}

Denote by ${T_{i{n_f}}}(n)$ the sweepers inward advancement time corresponding to the spread of evaders originating from the center of the evader region at the beginning of the last spiral sweep that had time of ${T_l}(n)$ to spread at a speed of $V_T$. Therefore, ${T_{i{n_f}}}(n)$ is given by ${T_{i{n_f}}}(n) = \frac{{{T_l}(n){V_T}}}{{{V_s}}}$. ${T_{in}}(n)$, the total inward advancement time therefore equals,
\begin{equation}
{T_{in}}(n) = \widetilde{T}_{in}(n) + {T_{_{in}last}}(n) + \eta {T_{i{n_f}}}(n)
\label{e717}
\end{equation}
Or alternatively, 
\begin{equation}
{T_{in}}(n) = \widetilde{T}_{in}(n) + \frac{{{R_{N}}}}{{{V_s}}} + \frac{{\eta r\left( {\lambda - 1} \right)}}{{{V_s}}}
\label{e718}
\end{equation}
\end{proof}

\subsection{Experimental Results}

Fig. $3$ presents a numerical analysis of the total sweep times required to detect all evaders in the region for different even numbers of sweepers ($2$ to $16$ sweepers). The total sweep times are the sum of spiral sweep and inward advancement times. 

\begin{figure}[htb!]
\noindent \centering{} \includegraphics[width=3.4in,height =3.4in]{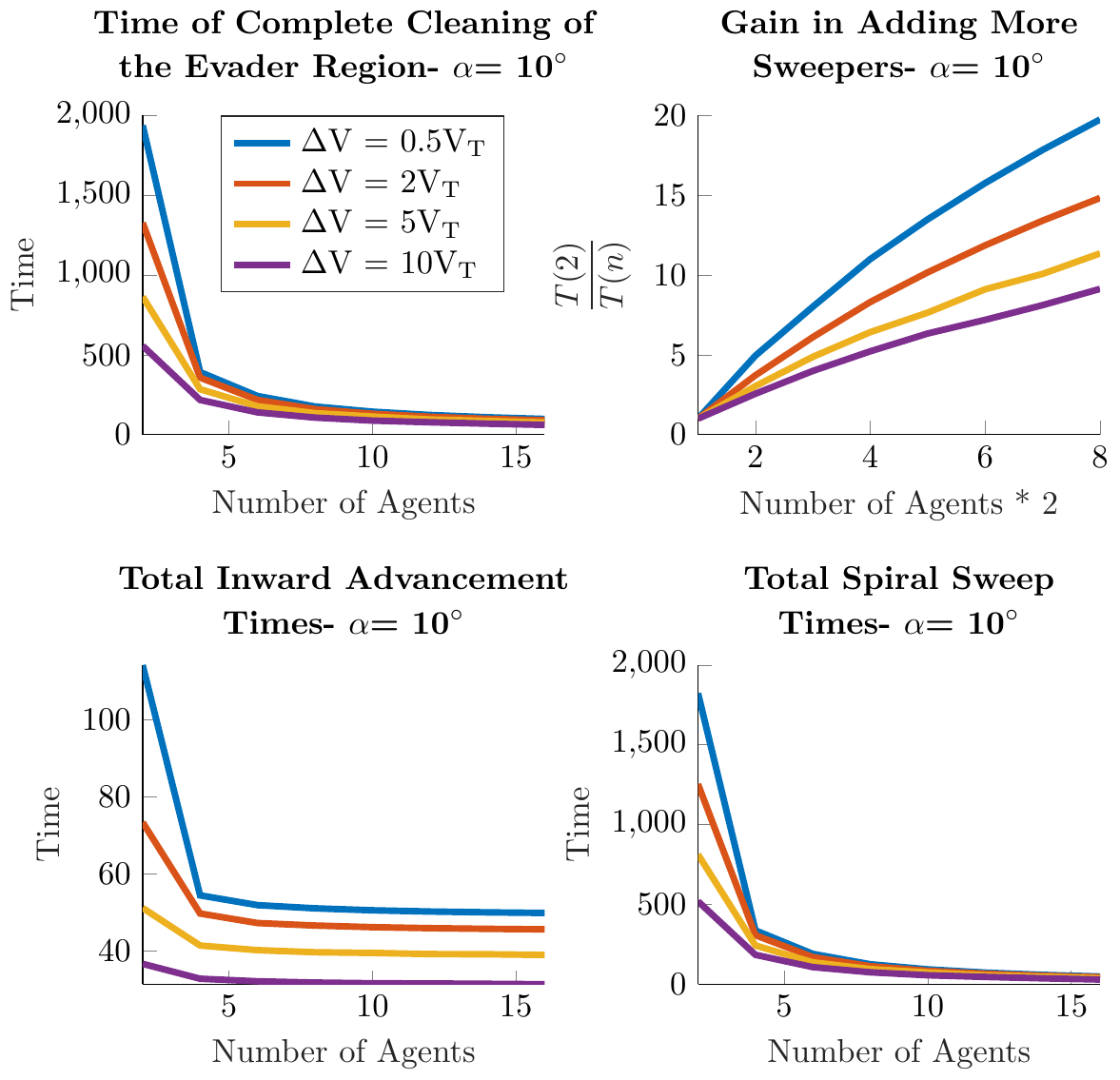} \caption{Sweep time until complete cleaning. The top left plot shows complete cleaning times for different numbers of sweepers. The top right plot shows the gain in search time reduction as a function of the number of sweepers. The bottom left plot shows only the spiral sweep time component and the bottom right plot shows only the inward advancement time component. The chosen values of the parameters are: $\alpha = 10 ^\circ$, $r=100$, $V_T = 1$ and $R_0 = 1000$.}
\label{Fig4Label}
\end{figure}

\section{Conclusions and Future Research Directions}

This research studies the problem of guaranteed detection of smart mobile evaders by a team of sweeping agents equipped with fan-shaped sensors that act as visual detectors. Evaders are initially located inside a known circular environment that does not have physical barriers that prevent evaders from attempting to escape it. An algorithm that guarantees detection of all evaders by using any even number of sweepers that use pincer sweeps between searching pairs is developed and analytically proven. Numerical and illustrative simulations using MATLAB and NetLogo demonstrate the performance of the proposed algorithm.

While in this work we focus on the rather simplistic circular environment, the concept of spiral pincer movements based on pairs of sweepers can be extended and generalized to more complex environments with different geometric layouts. Hence, a future research direction is to generalize the results to environments with different geometries. 



An additional future extension seeks to analyze critical speeds and sweep times obtained for teams of sweepers with fan-shaped sensors that employ same-direction sweeps. We expect such methods to result in degraded performance compared to the pincer-based protocols developed in this work. Such protocols will enable to precisely quantify the performance improvement achieved with the developed pincer-based search protocols compared to their same-direction alternatives. 
\bibliographystyle{IEEEtran}
\bibliography{pincer_same_direction_search_for_smart_evaderscomparison}
\end{document}